\documentclass[12pt]{amsart}
\usepackage{amssymb,amscd}
\usepackage{verbatim}

\usepackage{xcolor}

\usepackage{amsmath,amssymb,graphicx,mathrsfs}   
\usepackage[colorlinks=true,allcolors = blue]{hyperref} 

\usepackage{enumerate}

\usepackage{tikz}
\usetikzlibrary{matrix}

\usepackage[all]{xy}

\let\frak\mathfrak

\def\>{\relax\ifmmode\mskip.666667\thinmuskip\relax\else\kern.111111em\fi}
\def\<{\relax\ifmmode\mskip-.333333\thinmuskip\relax\else\kern-.0555556em\fi}
\def\vsk#1>{\vskip#1\baselineskip}
\def\vv#1>{\vadjust{\vsk#1>}\ignorespaces}
\def\vvn#1>{\vadjust{\nobreak\vsk#1>\nobreak}\ignorespaces}

  \let\ssize\scriptstyle
\let\sssize\scriptscriptstyle

\let\Medskip\medskip
\def\medskip{\par\Medskip}
\let\Bigskip\bigskip
\def\bigskip{\par\Bigskip}

\let\Maketitle\maketitle
\def\maketitle{\Maketitle\thispagestyle{empty}\let\maketitle\empty}

\newtheorem{thm}{Theorem}[section]
\newtheorem{cor}[thm]{Corollary}
\newtheorem{lem}[thm]{Lemma}

\theoremstyle{definition}                                  

\numberwithin{equation}{section}

\theoremstyle{definition}

\let\mc\mathcal
\let\nc\newcommand

\let\ka\kappa
\let\la\lambda

\let\phi\varphi
\let\si\sigma

\let\der\partial

\let\ox\otimes
\let\Tilde\widetilde

\let\geq\geqslant

\let\leq\leqslant

\let\on\operatorname
\let\bi\bibitem
\let\bs\boldsymbol

\def\C{{\mathbb C}}
\def\Z{{\mathbb Z}}
\def\R{{\mathbb R}}

\def\F{{\mathbb F}}

\def\+#1{^{\{#1\}}}

\def\Gr{\on{Gr}}

\def\gl{\mathfrak{gl}}

\def\beq{\begin{equation}}
\def\eeq{\end{equation}}
\def\be{\begin{equation*}}
\def\ee{\end{equation*}}

\nc{\bea}{\begin{eqnarray*}}
\nc{\eea}{\end{eqnarray*}}
\nc{\bean}{\begin{eqnarray}}
\nc{\eean}{\end{eqnarray}}

\let\ga\gamma
\let\Ga\Gamma

\nc{\Il}{{\mc I_{\bs\la}}}
\nc{\bla}{{\bs\la}}
\nc{\Fla}{\F_\bla}
\nc{\tfl}{{T^*\Fla}}
\nc{\GL}{{GL_n(\C)}}
\nc{\GLC}{{GL_n(\C)\times\C^*}}

\let\sd s 

\def\ddk_#1{\kk_{#1}\<\>\frac\der{\der\<\>\kk_{#1}}}

\def\bul{\mathbin{\raise.2ex\hbox{$\sssize\bullet$}}}
\def\intt{\mathchoice
{\mathop{\raise.2ex\rlap{$\,\,\ssize\backslash$}{\intop}}\nolimits}
{\mathop{\raise.3ex\rlap{$\,\sssize\backslash$}{\intop}}\nolimits}
{\mathop{\raise.1ex\rlap{$\sssize\>\backslash$}{\intop}}\nolimits}
{\mathop{\rlap{$\sssize\<\>\backslash$}{\intop}}\nolimits}}

\let\kk q 
\let\cc c

\let\Ko K

\def\GZ/{Gelfand-Zetlin}
\def\KZ/{{\slshape KZ\/}}
\def\qKZ/{{\slshape qKZ\/}}
\def\XXX/{{\slshape XXX\/}}

\def\Sym{\on{Sym}}

\nc{\A}{{\mc C}}

\begin{document}

\hrule width0pt
\vsk->

\title[Calogero-Moser eigenfunctions modulo $p^s$]{
Calogero-Moser eigenfunctions modulo $p^s$}

\author
[Alexander Gorsky and Alexander Varchenko]
{Alexander Gorsky$\>^\circ$ and Alexander Varchenko$\>^\star$}

\maketitle

\begin{center}
{\it ${}^\circ$ Institute for Information Transmission Problems, Moscow, Russia
\\
Bolshoy Karetny per. 19, Moscow, 127051, Russia \/}

\vsk.5>
{\it $^{\star}\<$Department of Mathematics, University
of North Carolina at Chapel Hill
\\ Chapel Hill, NC 27599-3250, USA\/}

\end{center}

{\let\thefootnote\relax
\footnotetext{\vsk-.8>\noindent
$^\circ\<${\sl E\>-mail}:\enspace shuragor@mail.ru\>
\\
$^\star\<${\sl E\>-mail}:\enspace anv@email.unc.edu\>,
supported in part by NSF grant  DMS  1954266}}

\begin{abstract}
    In this note we use the Matsuo-Cherednik duality between the solutions to KZ equations
    and eigenfunctions of Calogero-Moser Hamiltonians to get the polynomial $p^s$-truncation of the
    Calogero-Moser eigenfunctions at a rational coupling constant. The truncation procedure 
    uses the integral representation for the hypergeometric solutions to KZ equations. 
    The $s\rightarrow \infty$ limit to the pure $p$-adic case has been analyzed in the $n=2$ case.
\end{abstract}

\bigskip

{\small \tableofcontents  }

\setcounter{footnote}{0}
\renewcommand{\thefootnote}{\arabic{footnote}}

\section{Introduction}

\subsection{KZ equations and Calogero-Moser systems}

There exists a mapping between solutions to the KZ equations and eigenfunctions of the Calogero-Moser systems \cite{Ma,Ch}. The symmetrizations and antisymmetrizations of  solutions to the 
KZ equation become the eigenfunctions
of the  Calogero-Moser operators with the corresponding coupling constants. This correspondence
can be generalized to the relations between the eigenfunctions of the Riujsenaars models and solutions to qKZ
equations. One can also consider the limiting correspondence  when the Calogero-Moser and Ruijsenaars systems
get degenerated to the open Toda and relativistic open Toda systems and simultaneously the corresponding 
limit is taken at the KZ and qKZ side \cite{ZZ1, KPSZ, KZ, GVZ}. 
This Matsuo-Cherednik duality is referred to as quantum-quantum duality within the integrability framework.

The semiclassical limit of this duality can also be examined. In this scenario, the saddle point approximation is employed in the integral representation of the solutions to the KZ equations, leading to the derivation of the Bethe anzatz equations and the corresponding Bethe eigenvectors. This approach is applicable to quantum spin models such as Gaudin, XXX, or XXZ spin chains.
The Bethe Ansatz equations on the quantum spin chain side are equivalent to the equations for critical points of the master function \cite{GiK, MTV1, MTV2}. On the Calogero-Moser side, instead of focusing on wave functions, we examine the equations for the intersection of Lagrangian submanifolds, which correspond to the Bethe Ansatz equations on the spin chain side \cite{GiK, MTV2, GZZ}. Within the integrability framework, this duality is commonly known as classical-quantum duality, as it establishes a connection between the classical Calogero-Moser system and the quantum spin chain.

There are also  dualities within the families, known as bispectral dualities. 
On the Calogero-Moser side, these dualities interchange the coordinates and the spectral variables, providing a symmetry of the eigenfunction $\Psi(x,\lambda)$ \cite{R,FGNR}. The eigenfunction $\Psi(x,\lambda)$ turns out to be an eigenfunction for two different Hamiltonians that are related through bispectral transformation. 
On the spin chain side, there is an interesting interchange between the inhomogeneities and the eigenvalues of the twist operator. Specifically, the solution to the KZ equations for the inhomogeneities serves as a simultaneous solution to the bispectrally dual, also known as the dynamical or dual KZ equations, for the twist variables \cite{FMTV, OS, TV, EV1, MMRZZ}. Notably, the bispectral dualities exhibit commutativity with both the Matsuo-Cherednik quantum-quantum duality and the classical-quantum duality \cite{GVZ}.

The classical-quantum and Matsuo-Cherednik  dualities can also be understood in 
terms of the quantum geometry for some quiver varieties, for example $T^{*}\Gr(k,n)$
for the trigonometric Ruijsenaars-Schneider model and the dual qKZ equations
or some of its degenerations.
For the quantum integrable  system the  relation with quantum geometry
has been uncovered for the first time for the non-relativistic Toda \cite{GiK} and
for relativistic Toda system in \cite{GL}. The $J$-functions for the partial flags,
which are the  counting functions for quantum cohomology and quantum K-theory
correspondingly,  turn out to be eigenfunctions of the quantum Toda chain
Hamiltonians\cite{GiK,GL}.
The similar relation has been found for Calogero-Moser system \cite{MTV2} and for 
different versions of the Ruijsenaar-Schneider system \cite{KPSZ,KZ}.
Semiclassically the quantum K-theory rings of cotangent bundles to flag varieties or flag
varieties themselves, can be regarded as algebras of functions on the Lagrangian subvarieties in the phase spaces of classical models.

On the opposite spin chain side the first step 
towards the relation with the quantum geometry has been done in \cite{NS1,NS2} where 
it was shown explicitly that the extremization equation for the twisted superpotential for $T^{*}\Gr(k,n)$
yields the Bethe Anzatz equation for the inhomogeneous XXZ spin chain. This observation
has been generalized to the statement that the vertex function for $T^{*}\Gr(k,n)$
counting the quasimaps to the quiver manifold introduced and discussed in \cite{O,AO}
obeys the qKZ equations. On the other hand the master function is identified with the 
twisted superpotential in the integral solutions to qKZ. The similar relations
were found for different degenerations of $T^{*}\Gr(k,n)$ \cite{GiK,GL, MTV2, KPSZ, KZ, GLO}.

Hence we have two related counting functions for the quiver variables - the $J$-function and
vertex function formulated in the different ways.  They are eigenfunctions of  the 
Calogero-Moser, Ruijsenaars-Schneyder or Toda differential operators  and 
simultaneously can be constructed from 
solutions to KZ or qKZ equations  \cite{OS,KPSZ,KZ,SmV2}.
To some extend this situation 
can be considered as a manifestation of the Matsuo-Cherednik duality in terms of the quantum 
geometry. Note also that Matsuo-Cherednik duality can be considered
as the realization of the 3D mirror symmetry for the gauge theories 
on the brane worldvolumes. It is formulated in terms of the proper motion of the 
corresponding brane configurations taking into account the rules of 
brane intersections  \cite{GK}.

In this study, we are motivated by the recent derivation of the $p^s$-truncated hypergeometric 
solutions for KZ equations \cite{SV2, EV2}. We employ the Matsuo-Cherednik duality, 
which connects the solutions to KZ equations with the eigenfunctions of 
Calogero-Moser Hamiltonians. Our aim is to obtain the polynomial 
$p^s$-truncation of the Calogero-Moser eigenfunctions at a rational coupling constant. 
To accomplish this, we utilize the integral representation for the hypergeometric solutions to the KZ equation. Furthermore, we analyze the $s\to \infty$ limit, specifically in the case where $n=2$, to examine the pure $p$-adic scenario.

\subsection{Exposition of material}

Let $p$ be an odd prime number and $n, s$ positive integers, $n\geq 2$.
In this paper  we construct polynomials $\phi_{s,d}(z,\la,\ka)$ and $\psi_{s,d}(z,\la,\ka)$ in
variables  $z=(z_1,\dots,z_n)$,  $\la=(\la_1,\dots,\la_n)$, depending on a rational number  $\ka$ and
a discrete parameter $d$. 
 The polynomials
are eigenfunctions modulo $p^s$ of the Calogero-Moser operators,
\bea
\Big( - \Delta + \ka(\ka-1) \sum_{i\ne j}\frac 1{(z_i-z_j)^2}\Big) \phi \equiv p^2E \phi\, \pmod{p^s},
\\
\Big( - \Delta + \ka(\ka+1) \sum_{i\ne j}\frac 1{(z_i-z_j)^2}\Big) \psi \equiv p^2E \psi\,  \pmod{p^s},
\eea
where  $\Delta = \der_1^2+\dots + \der_n^2$, $\der_i  = \frac{\der}{\der z_i}$, $i=1,\dots, n$,
$E=-(\la_1^2+\dots+\la_n^2)$, see Corollary \ref{cor CM}.

\vsk.2>

The polynomials $\psi_{s,d}$ and $\psi_{s,d}$ are derived using the Matsuo-Cherednik procedure applied to polynomial solutions modulo $p^s$ of the KZ equations. Such solutions modulo $p^s$ 
have been constructed recently in \cite{EV2,SV2}. 
The original Matsuo-Cherednik procedure in \cite{Ma,Ch}
begins with a solution $I(z,\la,\ka)$ of appropriate KZ equations and generates two eigenfunctions $\phi(z,\la,\ka)$ and $\psi(z,\la,\ka)$ of the Calogero-Moser operator  shown below:
\bea
\Big( - \Delta + \ka(\ka-1) \sum_{i\ne j}\frac 1{(z_i-z_j)^2}\Big) \phi &=& E \phi\,,
\\
\Big( - \Delta + \ka(\ka+1) \sum_{i\ne j}\frac 1{(z_i-z_j)^2}\Big) \psi 
&=& E \psi\,.
\eea
 The process for constructing the polynomials $\phi_{s,d}$ and $\psi_{s,d}$ is detailed in Section \ref{sec 2}. For a given
tuple $(p,n,\ka, s)$, this construction yields a finite set of polynomials $\phi_{s,d}$ and $\psi_{s,d}$.

\vsk.2>
In Theorem \ref{thm ch}, an important new integral formula is introduced for the eigenfunction $\psi$ in the Matsuo-Cherednik construction when applied to a hypergeometric integral solution of the KZ equations. 
The original Matsuo-Cherednik formula expresses $\psi$ as a summation of $n!$ terms, while the formula in \eqref{int ch} represents $\psi$ as a single integral involving a product of elementary functions. 
This formula bears resemblance to the vertex integral formulas found in \cite{FSTV, SmV1, SmV2}, which are associated with 
3D mirror symmetry. 
A $p^s$-analog of this formula can also be found in Corollary \ref{cor 2.8}.

\vsk.2>

We closely examine the case $n=2$ in Section \ref{sec 3}. In this situation,
 $z=(z_1,z_2)$,  $\la=(\la_1,\la_2)$. The parameter 
$d$ in $\phi_{s,d}(z,\la,\ka)$ and $\psi_{s,d}(z,\la,\ka)$ is a positive integer. We assume that 
\bean
\label{s1}
p = kq+1, \qquad
\ka= \frac r q +m\,, \qquad   \frac 12 < \frac rq < 1,
\eean
where $k,  q, r, m$ are  integers,  $k, r,q$ are positive and $r,q$ are relatively prime. 
We demonstrate that if $s$ is sufficiently large, then the polynomials $\phi_{s,d}$ and $\psi_{s,d}$ become zero
for $d>1$ and the polynomials $\psi_{s,1}$ and $\psi_{s,1}$
are nonzero, see Theorem \ref{thm d=1}.

The polynomials $\phi_{s,1}(z,\la,\ka) $  and 
$\psi_{s,1}(z,\la,\ka-1) $ 
are eigenfunctions modulo $p^s$ of the same Calogero-Moser operator with the same eigenvalue.
In  Theorem \ref{thm shift},  we show that  
\bean
\label{i shift}
(\la_2-\la_1)\,p\,\psi_{s,1}(z,\la,\ka-1)\, \equiv\,(\ka+m) \,\phi_{s,1}(z,\la,\ka) \pmod{p^s},
\eean
in other words, these eigenfunctions modulo $p^s$ are proportional. Thus, for sufficiently large $s$,
our construction give exactly one nonzero eigenfunction modulo $p^s$ up to a scalar factor. 

In Theorem \ref{thm shift}, we show that our eigenfunctions agree with the shift operator 
$D_\ka = \der_1 -\der_2  - \frac{2\ka}{z_1-z_2}$, namely,
\bea
D_\ka \psi_{s,d}(z,\la,\ka-1)\, \equiv\, \ka\,\psi_{s,d}(z,\la,\ka)\pmod{p^s}.
\eea

In Section \ref{sec 4}, we consider the polynomials 
$\phi_{s,1}(z,\la,\ka)$ for $\ka=\frac rq$ as in \eqref{s1} and study the $p$-adic limit 
of these polynomials as $s\to\infty$.
 After a suitable normalization,
we present the polynomials $\phi_{s,1}(z,\la,\ka)$ as polynomials 
$M^0_s(x)$ in one variable  $x=(\la_2-\la_1)(z_1-z_2)$. In Theorem \ref{thm lim}, we show that 
the sequence of polynomials $\{M^0_s(x)\}$ uniformly $p$-adically 
converges on $\Z_p$ as $s\to\infty$ to the function
\bea
M(x)\  =\
\sum _{{n=0}}^{\infty }\frac {\big(-\frac{r}{q}\big)_n}{\big(-\frac{2r}{q}\big)_n\,n!}\,p^nx^n  
\eea
analytic on $\Z_p$.  This function $M(x)$ is a special case of Kummer's (confluent hypergeometric) function $M(a, b; x)$, 
\bea
 M(a,b;x)=
    {\frac {\Gamma (b)}{\Gamma (a)\Gamma (b-a)}}
  \sum _{{n=0}}^{\infty }\frac {\Ga(a+n)x^{n}}{\Ga(b+n)n!}  
  =\sum _{{n=0}}^{\infty }\frac {(a)_nx^{n}}{(b)_nn!}
   =     {}_{1}F_{1}(a;b;x),
\eea
where $(a)_n = a(a+1)\dots(a+n-1)$.

In Section \ref{sec 5.1} we discuss a parallel between the roles of the reciprocal 
$1/\hbar$ of the Planck constant in the conventional quantum mechanics and the prime number $p$ in our 
$p^s$-approxima\-tions of the eigenfunctions of the Calogero-Moser operator.  
In Section \ref{sec 5.2} we consider  the quasiclassical limit modulo $p^s$ of the Matsuo-Cherednik duality
for $n=2$.  In Section \ref{sec 6} we discuss some open problems.

\subsection*{Acknowledgments} 

The authors are grateful to A.\,Adolphson and S.\,Sperber for their
consultation on the inequality \eqref{SA}.
The authors thank IHES for hospitality in summer 2023.

\section{Hypergeometric integrals}
\label{sec 2}

\subsection{Master function} Fix a positive integer $n$.
Let $\la=(\la_1,\dots,\la_n)$,   $z=(z_1,\dots,z_n)$,
\bea
t=(t_{i,j}),\qquad i=1,\dots, n-1,\quad  j= 1,\dots, n-i.
\eea
The total number of variables $t_{i,j}$ equals $\frac{n(n-1)}2$.
For example, for $n=2$, we have $t=(t_{1,1})$ and for
$n=3$,  $t=(t_{1,1}, t_{1,2}, t_{2,1})$.

\vsk.2>
Define the {\it master function} by the formula
\bea
\Phi(t,z,\la, \ka) 
&=&
e^{\la_1\sum_{a=1}^n z_a  + \sum_{i=1}^{n-1} (\la_{i+1}-\la_i)\sum_{j=1}^{n-i}t_{i,j}}
\prod_{1\leq a<b\leq n} (z_a-z_b)^\ka
\prod_{j=1}^{n-1}\prod_{a=1}^{n} (t_{1,j}-z_a)^{-\ka}
\\
&\times &
\prod_{i=1}^{n-2} \prod_{1\leq j< j' \leq n-i}(t_{i,j}-t_{i,j'})^{2\ka}
\prod_{i=1}^{n-2} \prod_{j=1}^{n-i} \prod_{j'=1}^{n-i-1} (t_{i,j}-t_{i+1,j'})^{-\ka}.
\eea

\subsection{Symmetrization}

For a function $f(x_1,\dots,x_k)$ define
\bea
\Sym_{x_1,\dots,x_k} f = \sum_{\si\in S_k} f(x_{\si(1)}, \dots, x_{\si(k)}).
\eea

\subsection{Weight functions}

Consider $\C^n$ as the $\gl_n$-module of highest weight $(1,0,\dots,0)$ with
 the standard weight basis $v_1,\dots, v_n$. We have $e_{i,i}v_j=\delta_{i,j}v_j$.
  Let
 \bea
 (\C^n)^{\ox n}[1,\dots,1] \ \subset \ (\C^n)^{\ox n}
 \eea 
 be the weight subspace of $(\C^n)^{\ox n}$ of weight $(1,\dots,1)$. A basis of 
 $(\C^n)^{\ox n}[1,\dots,1]$ is given by the vectors
 $v_\si$, $\si \in S_n$, where
 \bea
 v_\si = v_{\si(1)}\ox\dots\ox v_{\si(n)}\,.
 \eea
 For every $\si\in S_n$ we construct a scalar weight function $V_\si(t,z)$ as follows. Define
 \bea
 V^\circ (t,z) =\prod_{a=1}^{n-1} 
\frac 1{(t_{1,n-a} -z_{a+1}) (t_{2,n-a}-t_{1,n-a})(t_{3,n-a}-t_{2,n-a})\dots
(t_{a,n-a}-t_{a-1,n-a})}\,,
\eea
\bea
V(t,z) = \Sym_{t_{1,1},\dots,t_{1,n-1}}
\Sym_{t_{2,1},\dots,t_{2,n-2}}\dots \Sym_{t_{n-2,1},t_{n-2,2}} V^\circ(t,z).
\eea
For example, for $n=2$,  
\bea
V(t_{1,1},z_1,z_2) = \frac 1{t_{1,1}-z_2}\,,
\eea
for $n=3$,
\bea
V^\circ(t_{1,1},t_{1,2}, t_{2,1}, z_1,z_2,z_3) \,
&=& \,\frac 1{t_{1,2}-z_2 } \,\frac 1{(t_{1,1}-z_3)(t_{2,1}-t_{1,1})}\,.
\\
V(t_{1,1},t_{1,2}, t_{2,1}, z_1,z_2,z_3) \,
&=& \,\frac 1{t_{1,2}-z_2 } \,\frac 1{(t_{1,1}-z_3)(t_{2,1}-t_{1,1})}\,
\\
&+&
 \,\frac 1{t_{1,1}-z_2 } \,\frac 1{(t_{1,2}-z_3)(t_{2,1}-t_{1,2})}\,.
\eea
For $\si\in S_n$, we define
\bea
V_\si (t,z_1,\dots,z_n) =  V(t, z_{\si^{-1}(1)},\dots, z_{\si^{-1}(n)}),
\eea
and
\bean
\label{int}
I^{(\ga)}_\si (z,\la,\ka) = \int_{\ga} \Phi(t,z,\la,\ka)  V_\si(t,z) dt\,
\eean
where $\ga$ is an $\frac{n(n-1)}2$-dimensional cycle and 
\bea
dt = \wedge_{i,j}dt_{i,j}
\eea
is the product in the lexicographical order of all differentials $dt_{i,j}$.

\subsection{Hypergeometric solutions}

Denote
\bea
\der_i
=
\frac{\der}{\der z_i}, \quad i=1,\dots, n, \quad\quad
\Delta &=& \der_1^2+\dots + \der_n^2.
\eea

\begin{thm} [\cite{SV1, FMTV}]
\label{thm FMTV}

The vector 
\bea
I^{(\ga)}(z,\la,\ka) =\sum_{\si\in S_n} I^{(\ga)}_\si(z,\la,\ka) v_\si
\eea
 is a solutions of the KZ equations,
\bean
\label{kz}
\der_i I = \Big(\ka \sum_{j\ne i} \frac{P^{(i,j)}}{z_i-z_j} +\la^{(i)}\Big) I, \qquad i=1,\dots,n.
\eean
The $\la^{(i)}$ is the operator acting as a diagonal matrix in the i-th factor and identically on all
other factors in the tensor product.

\end{thm}

\begin{thm} [\cite{Ma,Ch}]
\label{thm MC}
Let $I=\sum_{\si\in S_n} I_\si v_\si$ be a solution of the KZ equations. 
Define
\bea
\on{m}(I) = \sum_{\si\in S_n} I_\si\,, \qquad
\on{ch}(I) = \sum_{\si\in S_n} (-1)^{\si} I_\si\,.
\eea
Then $\phi = \on{m}(I)$ and $\psi = \on{ch}(I)$ are eigenfunctions of 
the Calogero-Moser operators,
\bean
\label{CMphi}
\Big( - \Delta + \ka(\ka-1) \sum_{i\ne j}\frac 1{(z_i-z_j)^2}\Big) \phi = E \phi\,,
\eean
\bean
\label{CMpsi}
\Big( - \Delta + \ka(\ka+1) \sum_{i\ne j}\frac 1{(z_i-z_j)^2}\Big) \psi = E \psi\,,
\eean
where $E=-(\la_1^2+\dots+\la_n^2)$.
\end{thm}

\subsection{An integral formula for $\on{ch}(I^{(\ga)})$}

\begin{thm}
\label{thm ch}
We have
\bean
\label{int ch}
\on{ch}(I^{(\ga)})(z,\la,\ka) = \int_\ga \Phi(t,z,\la, \ka+1) dt.
\eean

\end{thm}

This formula is analogous to the vertex integral formulas in \cite{FSTV, SmV1, SmV2}.

\begin{proof}

Formula \eqref{int ch} is equivalent to the formula
\bean
\label{ant wf}
\sum_{\si\in S_n} (-1)^\si V_\si (t,z) 
=
\frac{\prod_{1\leq a<b\leq n} (z_a-z_b) \prod_{i=1}^{n-2} \prod_{1\leq j< j' \leq n-i}(t_{i,j}-t_{i,j'})^{2}}
{\prod_{j=1}^{n-1}\prod_{a=1}^{n} (t_{1,j}-z_a)\prod_{i=1}^{n-2} \prod_{j=1}^{n-i} \prod_{j'=1}^{n-i-1} (t_{i,j}-t_{i+1,j'})}.
\eean

For example, for $n=3$ we have
\bea
&&
\frac1{(t_{1,2}-z_2)(t_{1,1}-z_1)(t_{2,1}-t_{1,1})} + 
\frac1{(t_{1,1}-z_2)(t_{1,2}-z_1)(t_{2,1}-t_{1,2})} 
\\
&-&
\frac1{(t_{1,2}-z_1)(t_{1,1}-z_2)(t_{2,1}-t_{1,1})} + 
\frac1{(t_{1,1}-z_1)(t_{1,2}-z_2)(t_{2,1}-t_{1,2})} 
\\
&-&
\frac1{(t_{1,2}-z_2)(t_{1,1}-z_3)(t_{2,1}-t_{1,1})} + 
\frac1{(t_{1,1}-z_2)(t_{1,2}-z_3)(t_{2,1}-t_{1,2})} 
\\
&-&
\frac1{(t_{1,2}-z_3)(t_{1,1}-z_1)(t_{2,1}-t_{1,1})} + 
\frac1{(t_{1,1}-z_3)(t_{1,2}-z_1)(t_{2,1}-t_{1,2})} 
\\
&+&
\frac1{(t_{1,2}-z_3)(t_{1,1}-z_2)(t_{2,1}-t_{1,1})} + 
\frac1{(t_{1,1}-z_3)(t_{1,2}-z_2)(t_{2,1}-t_{1,2})} 
\\
&-&
\frac1{(t_{1,2}-z_1)(t_{1,1}-z_3)(t_{2,1}-t_{1,1})} + 
\frac1{(t_{1,1}-z_1)(t_{1,2}-z_3)(t_{2,1}-t_{1,2})} 
\\
&=& \frac{(t_{1,1}-t_{1,2})^2 \prod_{1\leq a<b\leq 3} (z_a-z_b)}
{(t_{2,1}-t_{1,1})(t_{2,1}-t_{1,2})
\prod_{j=1}^{2}\prod_{a=1}^{3} (t_{1,j}-z_a)}\,.
\eea

Denote the left-hand in \eqref{ant wf} by $L$. It is as a ratio of two polynomials $P/Q$,
where each of them is a homogeneous polynomial in variables $t,z$. We may write
\bea
Q=\prod_{j=1}^{n-1}\prod_{a=1}^{n} (t_{1,j}-z_a)\prod_{i=1}^{n-2} \prod_{j=1}^{n-i} \prod_{j'=1}^{n-i-1} (t_{i,j}-t_{i+1,j'})
\eea
and $\deg Q-\deg P=\frac{n(n-1)}2$. It is easy to see that $L$ has the following properties:
\begin{enumerate}
\item
 for every $i=1,\dots, n-2$, the left-hand side $L$ is symmetric with respect to permutations of the variables
$t_{i,1},\dots,  t_{i,n-i}$;

\item  if $t_{i,j}=t_{i,j'}$ for some $i,j,j'$ with $j\ne j'$, then $L$ equals zero;

\item $L$ is anti-symmetric with respect to permutations of the variables $z_1,\dots, z_n$.
\end{enumerate}
These remarks imply that
\bea
P= C \prod_{1\leq a<b\leq n} (z_a-z_b) \prod_{i=1}^{n-2} \prod_{1\leq j< j' \leq n-i}(t_{i,j}-t_{i,j'})^{2}
\eea
where $C$ is a constant. We observe that $C=1$ by calculating any iterated residue of both sides in \eqref{ant wf}. 
The theorem is proved.
\end{proof}

\subsection{Reformulation of the equations}

Change the variables $\la\mapsto p\la$ where $p$ is an odd prime number\footnote{One can make a more general change of variables $\la\to p^r\la$ where $r>1/(p-1)$. Then all the results below remain true for this more general rescaling, cf. \cite{EV2}.}.  Then the KZ equations \eqref{kz} take the form 
\bean
\label{kzP}
\der_i I = \Big(\ka \sum_{j\ne i} \frac{P^{(i,j)}}{z_i-z_j} +p\la^{(i)}\Big) I, \qquad i=1,\dots,n,
\eean
and the Calogero-Moser equations 
\eqref{CMphi} and \eqref{CMpsi} take the form
\bean
\label{CMphip}
\Big( - \Delta + \ka(\ka-1) \sum_{i\ne j}\frac 1{(z_i-z_j)^2}\Big) \phi = p^2E \phi\,,
\eean
\bean
\label{CMpsip}
\Big( - \Delta + \ka(\ka+1) \sum_{i\ne j}\frac 1{(z_i-z_j)^2}\Big) \psi = p^2E \psi\,.
\eean

\subsection{Truncated exponential function}

Let $p$ be an odd prime number, $s$ a positive integer.
Denote
\bean
\label{s e}
d(s) = \Big[s\frac{p-1}{p-2}\Big]+1,
\qquad
E_s(x) = \sum_{k=0}^{d(s)} \,\frac{x^k}{k!}\,.
\eean

\begin{lem}
[\cite{EV2}]

\label{lem ps}

 If $\mu\in \Z_p$, then \ $e^{p\mu x}\in \Z_p[[x]],$\ \ $E_s(p\mu x)\in \Z_p[x]$\  and
 \bean
\notag
e^{p\mu x} &\equiv& E_s(p\mu x) \pmod{p^s},
\\
 \label{eE}
\frac{\der}{\der x}E_s(p\mu x) 
&\equiv&
 p\mu E_s(p\mu x), \qquad
\frac{\der}{\der \mu }E_s(p\mu x) \equiv px E_s(p\mu x)
 \pmod{p^s},
 \\
 \label{u+v}
 E_s(p\mu (u+v)) &\equiv&  E_s(p\mu u)  E_s(p\mu v)
 \pmod{p^s}. 
\eean

\end{lem}

\subsection{The $p^s$-approximations} 
\label{sec 2.8}

Let 
\bean\label{ass1}
\ka= \frac r q\,, \qquad   \frac 12 < \frac rq < 1,
\eean
where $r, q$ are relatively prime positive integers. Hence $q\geq 3$.  Let $p$ be a prime number   of the form
\bean
\label{ass2}
p = kq+1
\eean
 for some positive integer $k$. Let $s$ be a positive integer.

Define 
 \bea
\mc E_s(t,z,\la) = E_s\Big(
p\la_1\sum_{a=1}^n z_a  + \sum_{i=1}^{n-1} p(\la_{i+1}-\la_i)\sum_{j=1}^{n-i}t_{i,j}\Big).
\eea
Define the {\it master polynomial}  $\Phi_s(t,z,\la,\ka)$ by the formula
\bean
\label{mast pol}
\phantom{aaa}
\Phi_s(t,z,\la,\ka) &=& \mc E_s(t,z,\la) 
\prod_{1\leq a<b\leq n} (z_a-z_b)^{\frac{(q-r)p^s+r}q}
\prod_{j=1}^{n-1}\prod_{a=1}^{n} (t_{1,j}-z_a)^{r\frac{p^s-1}q}
\\
\notag
&
\times &
\prod_{i=1}^{n-2} \prod_{1\leq j< j' \leq n-i}(t_{i,j}-t_{i,j'})^{2\frac{(q-r)p^s+r}q}
\prod_{i=1}^{n-2} \prod_{j=1}^{n-i} \prod_{j'=1}^{n-i-1} (t_{i,j}-t_{i+1,j'})^{r\frac{p^s-1}q}.
\eean
Set
\bea
\Phi_{s,\si}(t,z,\la,\ka) = \Phi_s(t,z,\la,\ka) V_\si(t,z).
\eea
Note that $\Phi_{s,\si}(t,z,\la,\ka)$ is a polynomial in $t,z,\la$.

Let $d=(d_{i,j})$ be a vector of positive integers.
Denote by
\bea
I_{s,\si, d}(z,\la,\ka)
\eea
the coefficient of the monomial
$\prod_{i,j} t_{i,j}^{d_{i,j}p^s-1}$ in the polynomial $\Phi_{s,\si}(t,z,\la,\ka)$.

\begin{thm}
[\cite{EV2,SV2}]
\label{thm}
For any $d$,  the vector
\bea
I_{s,d}(z,\la,\ka) = \sum_{\si\in S_n} I_{s,\si,d}(z,\la,\ka) v_\si
\eea
is a solution modulo $p^s$
of the KZ equations \eqref{kzP} with $\ka=\frac rq$,
\bean
\label{kzp}
\der_i I \equiv \Big(\ka \sum_{j\ne i} \frac{P^{(i,j)}}{z_i-z_j} +p\la^{(i)}\Big) I \pmod{p^s}, \qquad i=1,\dots,n.
\eean

\end{thm}

\begin{cor}
\label{cor CM}
The polynomials $\phi_{s,d}=\on{m}(I_{s,d})$ and $\psi_{s,d}=\on{ch}(I_{s,d})$ are eigenfunctions modulo $p^s$
of the corresponding Calogero-Moser operators \eqref{CMphip} and \eqref{CMpsip} with $\ka=\frac rq$\,,
\bean
\label{CMp}
\Big( - \Delta + \ka(\ka-1) \sum_{i\ne j}\frac 1{(z_i-z_j)^2}\Big) \phi_{s,d}(z,\lambda,\ka) \equiv p^2E \,\phi_{s,d}(z,\lambda,\ka) \pmod{p^s},
\eean
\bean
\label{CMp2}
\Big( - \Delta + \ka(\ka+1) \sum_{i\ne j}\frac 1{(z_i-z_j)^2}\Big) \psi_{s,d} (z,\lambda,\ka)
\equiv p^2E \,\psi_{s,d}(z,\lambda,\ka) \pmod{p^s},
\eean
where $E=-(\la_1^2+\dots+\la_n^2)$ .

\end{cor}

Notice that $\Phi_{s,\si}(t,z,\la,\ka)$ is a polynomial in $t$ of degree independent on $\si$.
 Given $s$, there are at most finitely many vectors $d=(d_{i,j})$ which produce
nonzero vectors $I_{s,d}(z,\la,\ka)$. See below some remarks on the $n=2$ case.

\subsection{The $p^s$-approximations for $\ka=\frac rq +m$, $m\in \Z$}

In Section \ref{sec 2.8} we assumed that $p=kq+1$ and $\ka= \frac r q\,,$ where
 $\frac 12 < \frac rq < 1$, see \eqref{ass1} and \eqref{ass2}, and
 constructed solutions  modulo $p^s$  of the KZ equations and 
  eigenfunctions modulo $p^s$ of the Calogero-Moser operators. 
We may modify this construction under the assumptions that 
\bean
\label{ass3}
\ka= \frac r q + m\,, \qquad   \frac 12 < \frac rq < 1, \qquad p=kq+1,
\eean
where $r, q$ are relatively prime positive integers and $m$ is an integer. 

Let $s$ be a positive integer. 
Under assumption \eqref{ass3} we define  the {\it master polynomial}  $\Phi_s(t,z,\la,\ka)$ by the formula
\bean
\label{mast pol+m}
&&
\\
\notag
\Phi_s(t,z,\la,\ka) &=& \mc E_s(t,z,\la) 
\prod_{1\leq a<b\leq n} (z_a-z_b)^{\frac{(q-r)p^s+r}q +m}
\prod_{j=1}^{n-1}\prod_{a=1}^{n} (t_{1,j}-z_a)^{r\frac{p^s-1}q-m}
\\
\notag
&
\times &
\prod_{i=1}^{n-2} \prod_{1\leq j< j' \leq n-i}(t_{i,j}-t_{i,j'})^{2\frac{(q-r)p^s+r}q +2m}
\prod_{i=1}^{n-2} \prod_{j=1}^{n-i} \prod_{j'=1}^{n-i-1} (t_{i,j}-t_{i+1,j'})^{r\frac{p^s-1}q-m}.
\eean
Set
\bea
\Phi_{s,\si}(t,z,\la,\ka) = \Phi_s(t,z,\la,\ka) V_\si(t,z).
\eea
We assume that $s$ is such that
\bean
\label{p-ine}
p^s + (1-p^s)\frac r q +m > 0\qquad\on{and} \qquad 
(p^s-1)\frac r q-m> 0.
\eean
Then the functions $\Phi_{s,\si}(t,z,\la,\ka) $ are polynomials.

Let $d=(d_{i,j})$ be a vector of positive integers.
Denote by  $I_{s,\si, d}(z,\la,\ka)$ the coefficient of the monomial
$\prod_{i,j} t_{i,j}^{d_{i,j}p^s-1}$ in the polynomial $\Phi_{s,\si}(t,z,\la,\ka)$.

\begin{thm}
\label{thm m}

Under  assumptions \eqref{ass3} and \eqref{p-ine},  for  any positive integer $d$ 
the vector $I_{s,d}(z,\la,\ka) = \sum_{\si\in S_n} I_{s,\si,d}(z,\la,\ka) v_\si$
is a solution modulo $p^s$
of the KZ equations \eqref{kzP} with $\ka=\frac rq+m$, and
the polynomials
$\phi_{s,d}(z,\la,\ka)  = \sum_{\si\in S_n} I_{s,\si, d}(z,\la,\ka)$ 
and
$\psi_{s,d}(z,\la,\ka) = 
 \sum_{\si\in S_n} (-1)^\si I_{s,\si,d}(z,\la,\ka) $
are eigenfunctions modulo $p^s$
of the  Calogero-Moser operators \eqref{CMp} and \eqref{CMp2}, respectively,
with $\ka=\frac rq + m$.
\end{thm}

The proof of this theorem is the same as the proof of Theorem \ref{thm}.

\vsk.2>

We also have the following corollary of  Theorem \ref{thm ch}.

\begin{cor}
\label{cor 2.8}
The polynomial  $\psi_{s,d}(z,\la,\ka)$ equals the coefficient of  the monomial
\\
$\prod_{i,j} t_{i,j}^{d_{i,j}p^s-1}$ in the polynomial $\Phi_{s}(t,z,\la,\ka+1)$.
\end{cor}

\section{The case $n=2$}
\label{sec 3}

\subsection{Hypergeometric solutions}

For $n=2$, we have $\la=(\la_1,\la_2)$, $z=(z_1,z_2)$. The hypergeometric solutions
of the KZ equations \eqref{kzp} 
involve a single integration variable
$t=t_{1,1}$.  The master function is 
\bean
\label{MF2}
\Phi(t,z,p\la, \ka) 
=
e^{p\la_1(z_1+z_2)  + p(\la_{2}-\la_1)t} (z_1-z_2)^\ka
(t-z_1)^{-\ka}(t-z_2)^{-\ka}.
\eean
The hypergeometric solutions have 
two coordinates,
\bea
I^{(\ga)}(z, p\la,\ka) = \left(I^{(\ga)}_1(z, p\la,\ka), \, I^{(\ga)}_2(z, p\la,\ka)\right),
\eea
where
\bean
\label{Ii}
I^{(\ga)}_i(z, p\la,\ka)
=\int_\ga \Phi(t,z,p\la, \ka) \frac{dt}{t-z_i}\,,\qquad i=1,2.
\eean
Then
\bea 
\phi^{(\ga)} = I^{(\ga)}_1(z, p\la,\ka)+ I^{(\ga)}_2(z, p\la,\ka),
\quad\psi^{(\ga)} = I^{(\ga)}_1(z, p\la,\ka)- I^{(\ga)}_2(z, p\la,\ka)
\eea
are eigenfunctions of the Calogero-Moser operators,
\bean
\label{CMphi2}
\Big( - \Delta + \frac{2\ka(\ka-1)}{(z_1-z_2)^2}\Big) \phi = p^2E \phi\,,
\eean
\bean
\label{CMpsi2}
\Big(-\Delta + \frac{2\ka(\ka+1)}{(z_1-z_2)^2}\Big) \psi = p^2E \psi\,,
\eean
where $\Delta=\der_1^2+\der_2^2$, \ 
$E=-\la_1^2-\la_2^2$.

\subsection{The $p^s$-approximations}

Let $\ka=\frac rq +m$ satisfies assumptions \eqref{ass3}.
Define
\bea
\mc E_s(t,z,\la)
&=&
E_s\big(
p\la_1(z_1+z_2)  +  p(\la_{2}-\la_1)t\big),
\\
\Phi_s(t,z,\la,\ka) 
&=& 
\mc E_s(t,z,\la)
(z_1-z_2)^{p^s + (1-p^s)\frac rq +m}
 \big((t-z_1)(t-z_2)\big)^{(p^s-1)\frac rq-m},
\\
\Phi_{s,i} (t,z,\la,\ka) &=&
 \Phi_s(t,z,\la,\ka) \frac1{t-z_i}\,,
\quad i=1,2.
\eea
Let $d$ be a  positive integer.
Denote by  $I_{s,i,d}(z,\la)$  the coefficient of the monomial
$t^{dp^s-1}$ in the polynomial $\Phi_{s,i} (t,z,\la,\ka)$.
Then for any positive integer $d$ the polynomials
\bean
\label{phipsi}
\phi_{s,d}(z,\la,\ka) 
&=& 
 I_{s,1,d}(z,\la,\ka) 
+ I_{s,2,d}(z,\la,\ka),
\\
\notag
\psi_{s,d}(z,\la,\ka) &=& 
 I_{s,1,d}(z,\la,\ka) 
- I_{s,2,d}(z,\la,\ka)
\eean
are eigenfunctions modulo $p^s$
of the  Calogero-Moser operators \eqref{CMphip} and \eqref{CMpsip} with $\ka=\frac rq+m$,
\bean
\label{CMphi2d}
\Big( - \Delta + \frac{2\ka(\ka-1)}{(z_1-z_2)^2}\Big) \phi_{s,d} \,\equiv\, p^2E \phi_{s,d}
\pmod{p^s},
\eean
\bean
\label{CMpsi2d}
\Big(-\Delta + \frac{2\ka(\ka+1)}{(z_1-z_2)^2}\Big) \psi_{s,d}\, \equiv\, p^2E \psi_{s,d}\,
\pmod{p^s}.
\eean

\begin{lem}
    \label{lem d=1 new ka}
If 
\bean 
\label{d=1}
d>1 \quad \on{and}
\quad p^s > \frac {s-m}{1-\frac rq}\,,
\eean
 then $\phi_{s,d}(z,\la,\ka) $
and $\psi_{s,d}(z,\la,\ka) $  are  zero polynomials.
\end{lem} 

\begin{proof}

The degree of $\Phi_{s,i} (t,z,\la)$ with repsect to $t$ is not greater
than
\bea
\Big[s\frac{p-1}{p-2}\Big]+1 +2 \frac rq (p^s-1)-2m -1\,,
\eea
where $[\cdot]$ denotes the integer part. We have
\bea
\Big[s\frac{p-1}{p-2}\Big]+1 + 2\frac rq(p^s-1)-2m-1 - (2p^s-1)
<
\Big(2\frac rq-2\Big) p^s  +2s -2m,
\eea
and $(2\frac rq-2) p^s  +2s-2m<0$  if the second inequality in \eqref{d=1} holds.    
\end{proof}

\begin{thm}
\label{thm d=1}
${}$

\begin{enumerate}
\item[$\on{(i)}$]
Let $\ka= \frac r q + m$ satisfy 
assumptions \eqref{ass3}. Let a positive integer $s$ satisfy  assumptions \eqref{p-ine}.
Then the 
$\phi_{s,1}(z,\la,\ka)$ is divisible by $\la_1-\la_2$ modulo $p^s$ and is zero  modulo $p$.
\item[$\on{(ii)}$]
Let $\ka= \frac r q$ and $p$ satisfy assumptions \eqref{ass1} and \eqref{ass2}.
Then the 
$\phi_{s,1}(z,\la,\ka)$ is  nonzero modulo $p^2$ and
the polynomial $\psi_{s,1}(z,\la,\ka)$ is  nonzero modulo $p$.
\end{enumerate}

\end{thm}

\begin{proof}
Let $\ka= \frac r q + m$ .
Let us prove that $\phi_{s,1}(z,\la,\ka)$
is zero modulo $p$. It is enough to prove that
\bean
\label{der+der}
\phi_{s,1}(z,0,\ka) = I_{s,1,1}(z,0,\ka)  +
I_{s,2,1}(z,0,\ka) \equiv
0\pmod{p^s} ,
\eean
since all powers of $\la$ come 
from $E_s\big(p\la_1(z_1+z_2)  +  p(\la_{2}-\la_1)t\big)$
with positive powers of $p$. The congruence \eqref{der+der} follows from the fact that the expression
\bea
\Big( (p^s-1)\frac rq-m\Big) \Big(I_{s,1,1}(z,0,\ka)  +
I_{s,2,1}(z,0,\ka)\Big)
\eea
equals the coefficient of $t^{p^s-1}$ in the polynomial $\frac{\der}{\der t}\Phi_s(t,z,0,\ka)$,
 and hence is congruent to zero modulo $p^s$.
 
 The fact that $\phi_{s,1}(z,\la, \ka)$ is divisible by $\la_1-\la_2$ modulo $p^s$ also follows from the congruence \eqref{der+der}. Part (i) of the theorem is proved.

\vsk.2>

Let $\ka=\frac rq$. Let us substitute $z=(z_1,z_2) = (0,-1)$, \, $\la=(\la_1,\la_2)=(0,0)$  to
$\psi_{s,1}$ and evaluate $C=\psi_{s,1}(0,1,0,0,\ka)$ modulo $p$. 
By Corollary \ref{cor 2.8} the number $C$ is the coefficient
of  $t^{p^s-1}$ in 
 $  t^{(p^s-1)\frac rq -1}(t+1)^{(p^s-1)\frac rq -1}$. 
 Hence,
$ C = \binom{A}{B}$, where 
$A=(p^s-1)\frac rq -1$, $B= (p^s-1)\frac {q-r}q +1$.

Recall  the assumptions:\  $\ka= \frac r q\,,$\  $ \frac 12 < \frac rq < 1$ and
$p = kq+1$. 
Then
\bea
A =rk (1+p+\dots +p^{s-1}) -1 ,
\qquad
B= (q-r)k(1+p+\dots+p^{s-1}) + 1,
\eea
where 
\bea
(q-r)k + 1\leq rk -1 < p-1.
\eea
Applying Lucas's theorem \cite{L}
we obtain
\bea
\binom{A}{B} &\equiv& \binom{rk-1}{(q-r)k+1} \binom{rk-1}{(q-r)k}^{s-1}\not\equiv 0 \pmod{p}.
\eea
Thus $\psi_{s,1}(z,\la,\ka)$ is not zero modulo $p$.

Assume that $\ka=\frac rq$ and prove that the polynomial
$\phi_{s,1} = I_{s,1,1}+ I_{s,1,1}$  is  nonzero modulo $p^2$. 
For that we substitute $z=(z_1,z_2) = (0,-1)$, \, $\la=(\la_1,\la_2)=(0,\la_2)$  to
$I_{s,1,1}$ and $I_{s,2,1}$  and obtain two polynomials in $\la_2$. We denote by
$D_i$ the coefficient of $\la_2$ in  $I_{s,i,1}(0,-1,0,\la_2)$ for $i=1,2$.
Then $D_1= p\binom{A+1}{B-1}$
and  $D_2= p\binom{A}{B-2}$.  
Similarly to the previous reasoning, Lucas's theorem implies that $D_1+D_2\not\equiv0\pmod{p^2}$.
Theorem \ref{thm d=1} is proved.
\end{proof}

\subsection{Polynomials  $\phi_{s,d}(z,\la,\ka) $
and $\psi_{s,d}(z,\la,\ka-1) $}

The polynomials $\phi_{s,d}(z,\la,\ka) $  and 
\\
$\psi_{s,d}(z,\la,\ka-1) $ are eigenfunctions modulo $p^s$ of the same Calogero-Moser operator with the same eigenvalue, 
\bea
\Big(-\Delta + \frac{2\ka(\ka-1)}{(z_1-z_2)^2}\Big) \mu\, \equiv\, p^2E \mu\,
\pmod{p^s}.
\eea
The next theorem shows that these  eigenfunctions 
are proportional modulo $p^s$.

\begin{thm}
\label{thm shift} 
If assumptions \eqref{ass3} and \eqref{p-ine} hold for $\ka=\frac r q + m$ and $\ka-1=\frac r q + m-1$, then
\bean
\label{shift}
(\la_2-\la_1)\,p\,\psi_{s,d}(z,\la,\ka-1)\, \equiv\,\ka \,\phi_{s,d}(z,\la,\ka) \pmod{p^s}.
\eean

\end{thm}

\begin{proof}
By Corollary \ref{cor 2.8}, the polynomial $\psi_{s,d}(z,\la,\ka-1)$ equals the coefficient of $t^{dp^s-1}$ in 
$\Phi_s(t,z,\la,\ka)$.
We also have
\bea
\frac\der{\der t} \Phi_s(t,z,\la,\ka) \equiv \Phi_s(t,z,\la,\ka)\Big((\la_2-\la_1)p
-\frac\ka{t-z_1} -\frac\ka{t-z_2}\Big) \pmod{p^s}.
\eea
This gives
 \eqref{shift}, since the coefficient of $t^{dp^s-1}$ in $\frac\der{\der t} \Phi_s(t,z,\la,\ka)$ is congruent to 0 modulo $p^s$.
\end{proof}

\subsection{Shift operator}

Denote 
\bea
L_\ka = -\Delta + \frac{2\ka(\ka-1)}{(z_1-z_2)^2}\,,
\qquad
D_\ka = \frac\der{\der z_1} - \frac\der{\der z_2}  - \frac{2\ka}{z_1-z_2}\,.
\eea
Then
\bea
L_{\ka+1}D_\ka = D_\ka L_\ka .
\eea
Hence, if $\phi$ is an eigenfunction of $L_\ka$ with  some eigenvalue, then $D_\ka\phi$ is 
an eigenfunction of $L_{\ka+1}$ with the same eigenvalue. 
The operator $D_\ka$ is called the shift operator, see for example \cite{FVe}.
\smallskip

For any positive integer $d$, the polynomial $\psi_{s,d}(z,\la,\ka-1)$ is an eigenfunction modulo $p^s$ of $L_\ka$ while 
 the polynomial $\psi_{s,d}(z,\la,\ka-1)$ is an eigenfunction modulo $p^s$ of $L_{\ka+1}$.

\begin{thm}
\label{thm shift}
We have
\bean
\label{shift e}
D_\ka \psi_{s,d}(z,\la,\ka-1)\, \equiv\, \ka\,\psi_{s,d}(z,\la,\ka)\pmod{p^s}.
\eean

\end{thm}

\begin{proof}
Direct calculation shows that
\bean
\label{D-the}
D_\ka \Phi_s(t,z,\la,\ka) \equiv \ka \Phi_s(t,z,\la,\ka+1) \pmod{p^s}.
\eean
Since $\psi_{s,d}(z,\la,\ka-1)$ is the coefficient of $t^{dp^s-1}$ in $\Phi_s(t,z,\la,\ka)$, 
and $\psi_{s,d}(z,\la,\ka)$ is the coefficient of $t^{dp^s-1}$ in $\Phi_s(t,z,\la,\ka+1)$,
we obtain \eqref{shift e}.
\end{proof}

\section{The $p$-adic limit as $s\to\infty$ for $n=2$}
\label{sec 4}

\subsection{Change of variables}

Recall the eigenfunction equation \eqref{CMpsi2}.
Its complex eigenfunctions  are given by the formula
\bean
\label{in psi}
\psi^{(\ga)}(z,\la,\ka) = \int_\ga e^{p\la_1(z_1+z_2)  + p(\la_{2}-\la_1)t} (z_1-z_2)^{\ka+1}
(t-z_1)^{-\ka-1}(t-z_2)^{-\ka-1}dt,
\eean
see Theorem \ref{thm ch}.
Change the integration variable  $t \to (z_1-z_2)t+z_2$, and choose $\ga$, so that this function becomes
\bean
\label{cint}
 e^{p(z_1\la_1+z_2\la_2)} (z_1-z_2)^{-\ka} \int_0^1 e^{ p(\la_2-\la_1)(z_1-z_2)  t} t^{-\ka-1}  (1-t)^{-\ka-1}dt.
\eean
The integral in \eqref{cint} is related to  Kummer's (confluent hypergeometric) function $M(a, b; x)$, 
\bean
\label{MT}
    M(a,b;x)=
    {\frac {\Gamma (b)}{\Gamma (a)\Gamma (b-a)}}
  \sum _{{n=0}}^{\infty }\frac {\Ga(a+n)x^{n}}{\Ga(b+n)n!}  
  =\sum _{{n=0}}^{\infty }\frac {(a)_nx^{n}}{(b)_nn!}
   =     {}_{1}F_{1}(a;b;x),
\eean
where $(a)_n = a(a+1)\dots(a+n-1)$. Kummer's function is a solution of Kummer's equation 
\bea
x{\frac {d^{2}w}{dx^{2}}}+(b-x){\frac {dw}{dx}}-aw=0
\eea
and has an integral representation
\bea
     M(a,b;x)={\frac {\Gamma (b)}{\Gamma (a)\Gamma (b-a)}}\int _{0}^{1}e^{{xt}}t^{{a-1}}(1-t)^{{b-a-1}}\,dt,
     \eea
if $\on{Re}\,b > \on{Re} a > 0$. Thus,
\bean
\label{Ku}
&&
\\
\notag
&&
\int_0^1 e^{p(\la_2-\la_1)(z_1-z_2) t} t^{-\ka-1}  (t+1)^{-\ka-1}dt = \frac{\Ga(-\ka)^2} {\Ga(-2\ka)}
M(-\ka,-2\ka; p(\la_2-\la_1)(z_1-z_2)).
\eean

\subsection{The $p$-adic limit}

We $p^s$-approximate the integral in the left-hand side of \eqref{Ku}
by polynomials as in the previous sections and show that the corresponding monic polynomials 
 $p$-adically tend to the power series
expansion in $p(\la_2-\la_1)(z_1-z_2)$
of the function  $M(-\ka,-2\ka; p(\la_2-\la_1)(z_1-z_2))$.
The transition from the $p^s$-approximating polynomials to the monic polynomials corresponds to
the prefactor $\frac{\Ga(-\ka)^2}{\Ga(-2\ka)}$ in \eqref{Ku}.

\vsk.2>

Notice 
that our $p^s$-approximation procedure does not require the use of the integration cycle $\gamma$. 
Therefore, the $p^s$-approximation procedure independently of any cycle identifies the solution to Kummer's equation 
that is analytic at $x=0$.

\vsk.2>

 Denote $x=(\la_2-\la_1)(z_1-z_2)$ and  
\bean
\label{Ku1}
M(x)&:= &M(-\ka,-2\ka, px) =\frac{\Ga(-2\ka)}{\Ga(-\ka)^2}
\int_0^1 e^{pxt} t^{-\ka-1}  (t+1)^{-\ka-1}dt
\\
\notag
&=&
\sum _{{n=0}}^{\infty }\frac {(-\ka)_n}{(-2\ka)_nn!}\,p^nx^n  
= 1 +\frac{-\ka}{-2\ka}\,px + \frac{(-\ka)(-\ka+1)}{(-2\ka)(-2\ka+1) 2}\,p^2x^2 \,+\, \dots.
\eean
Denote
\bean
\label{c_n}
c_n = \frac {(-\ka)_n}{(-2\ka)_nn!}.
\eean

Let $ p= kq+1$ and $\ka = \frac rq$ satisfy \eqref{ass1} and \eqref{ass2}.
Denote by $M_s(x)$ the coefficient of $t^{p^s-1}$ in 
\bea
\Psi_s(t)\, =\, E_s(pxt)\, t^{(p^s-1)\frac rq-1}(1-t)^{(p^s-1)\frac rq-1}\,.
\eea
Then
\bean
\label{M_s}
M_s(x) = \sum_{n=0}^{d(s)} (-1)^{(p^s-1)\frac {q-r}q-n+1 }
\binom{(p^s-1)\frac rq-1}{(p^s-1)\frac {q-r}q-n+1}  p^nx^n.
\eean
Define the normalized $p^s$-approximating polynomial:
\bean
\label{Mo}
M^0_s(x):= \frac{M_s(x)}{M_s(0)} =  
\sum_{n=0}^{d(s)}  c_{s,n}\,
  p^nx^n, \qquad c_{s,n} = 
  \frac{\big(\frac{-r}q +p^s\frac{r-q}q \big)_n}{\big(\frac{-2r}q +p^s\frac{2r-q}q \big)_n\,n!}\,.
\eean

\begin{thm}
For any $n$ the coefficient $c_{s,n}$ $p$-adically tends to $c_n$ and $s\to \infty$.

\end{thm}

\begin{proof}
The statement follows since both $p^s\frac{2r-q}q$ and $p^s\frac{r-q}q$ $p$-adically
tend to zero as 
\\
$s\to\infty$.
\end{proof}

We have a stronger statement. 

\begin{thm}
\label{thm lim}

Let $ p= kq+1$ and $\ka = \frac rq$ satisfy \eqref{ass1} and \eqref{ass2}. 
Then the power series
\bean
\label{Ku2}
M(x)\  =\
\sum _{{n=0}}^{\infty }\frac {\big(-\frac{r}{q}\big)_n}{\big(-\frac{2r}{q}\big)_n\,n!}\,p^nx^n  
\eean
$p$-adically uniformly converges for all $x\in\Z_p$\,. The
sequence of polynomials $\big(M_s^0(x)\big)_{s\geq 0}$ uniformly $p$-adically converges on $\Z_p$
to the analytic function $M(x)$.

\end{thm}

\begin{proof} 
Recall that   $c_n=\frac {\big(-\frac{r}{q}\big)_n}{\big(-\frac{2r}{q}\big)_n\,n!}$.
 Let $\on{ord}_p a$ denote the $p$-adic order of $a$. 
 It follows from   \cite[Section 1]{Dw}, 
also from \cite[Section 5]{AS}, that 
\bean
\label{SA}
\on{ord}_p r_n  \,\geq \,-\,\frac n{p-1} \,-\, \log_p n\,-\,1.
\eean
Hence
\bean
\label{cp}
\on{ord}_p (c_np^n)  \,\geq \,n\,\frac {p-2}{p-1} \,-\, \log_p n\,-\,1.
\eean
Hence,  given $d>0$, there exists $N_0>0$ such that
\bean 
\label{d1}
\on{ord}_p \left(\sum _{{n=N}}^{\infty } c_n p^nx^n\right) > {d}\,,
\eean
for all $N\geq N_0$, if $\on{ord}_p x\geq 0$. Thus, the series $M(x)$
$p$-adically uniformly converges for all $x\in\Z_p$\,.

The inequality \eqref{cp} also implies that there exists a real number $A$ such that
$\on{ord}_p (c_np^n)> A$ for all $n$.

\vsk.2>
Let us compare the coefficients $c_n$  and $c_{s,n}$. 
Under our assumptions we have the following $p$-adic presentations,
\bean
\label{pr}
\phantom{aaaaaaa}
\frac{-r}q = kr(1+p+p^2+\dots),
\qquad
\frac{-2r}q =k(2r-q)-1 + k(2r-q) (p+p^2+\dots),
\eean
where $0\leq k(2r-q)-1<p$. 

 The polynomial $M^0_s(x)$ is of degree $d(s) < 3s$ in $x$. It is clear from \eqref{pr} that if $n\leq d(s)$, then 
\bea
\frac{c_{s,n}}{ c_{n}} = 1 + \tilde c_{s,n} \quad \on{where}\quad \on{ord}_p \tilde c_{s,n} \geq s.
\eea
 Hence
\bea
\on{ord}_p(c_{s,n}- c_{n}) = \on{ord}_p \Big(c_n \Big(\frac{c_{s,n}}{ c_{n}} - 1\Big)\Big) = 
\on{ord}_p (c_n\tilde c_{s,n})\geq A+s,
\eea
and
\bean
\label{c2s}
\on{ord}_p \Big(M^0_s(x) - \sum _{n=0}^{d(s)} c_n p^nx^n\Big) \geq A+s
\eean
if $\on{ord}_p x\geq 0$.
The inequalities \eqref{d1} and \eqref{c2s} imply that the
sequence of polynomials $\big(M_s^0(x)\big)_{s\geq 0}$ uniformly $p$-adically converges on $\Z_p$
to the analytic function $M(x)$. 
\end{proof}

\section{Some properties of $p^s$-truncated Calogero-Moser model}
\label{sec 5}

\subsection{Arithmetic with the Planck constant}
\label{sec 5.1}

In previous sections we combined  a $p^s$-approximation of solutions to KZ equations  with an artificial rescaling 
$\la \rightarrow p\la$. Let us consider a  procedure which makes this rescaling natural.
To this aim introduce the Planck constant $\hbar$ in the
KZ equations and Calogero-Moser operator:
\bean
\label{hbar}
&&
\phantom{aaaa}
\hbar \der_i I  =  \Big(\ka \sum_{j\ne i} \frac{P^{(i,j)}}{z_i-z_j} +\la^{(i)}\Big) I , \qquad i=1,\dots,n,
\\
\notag
&&
\Big( -\hbar^2 \Delta + \ka(\ka+\hbar) \sum_{j\ne i}\frac 1{(z_i-z_j)^2}\Big)\Psi= E\Psi.
\eean
The Matsuo-Cherednik map holds at any  $\hbar$.  Let us divide the KZ equations (\ref{hbar}) by $\hbar$ and introduce 
$\tilde \ka =\ka/\hbar$. Then 
\bean
\label{hbar2}
\der_i I = \Big(\tilde{\ka} \sum_{j\ne i} \frac{P^{(i,j)}}{z_i-z_j} + \hbar^{-1}\la^{(i)}\Big) I , \qquad i=1,\dots,n.
\eean
These equations coincide with equations \eqref{kzP},
\bea
\der_i I = \Big(\ka \sum_{j\ne i} \frac{P^{(i,j)}}{z_i-z_j} +p\la^{(i)}\Big) I , \qquad i=1,\dots,n,
\eea
 upon the identification $\ka=\tilde \ka$ and  $p=\hbar^{-1}$.

For example, for $n=2$, the complex solutions to the KZ equations \eqref{hbar2} are given by periods of the master function
\bean
\label{MF3}
\Phi(t,z,\hbar^{-1}\la, \Tilde {\ka}) 
=
e^{\hbar^{-1}\la_1(z_1+z_2)  + \hbar^{-1}(\la_{2}-\la_1)t} (z_1-z_2)^{\tilde {\ka}}
(t-z_1)^{-\tilde {\ka}}(t-z_2)^{-\tilde {\ka}},
\eean
which is identical to the master function in  (\ref{MF2}) if we identify  $\ka=\tilde \ka$ and  $p=\hbar^{-1}$.

At the Calogero-Moser side, we also obtain the coupling constant $\tilde{\ka}$ and rescaled 
energy $E\rightarrow \hbar^{-2} E$. In this case, the conventional quantum mechanics 
associated with the master function $\Phi$ in \eqref{MF3} is present. Thus, we can question 
the physical meaning of our $\hbar^{-s}$-approximation procedure. Our approach differs from the standard small $\hbar$
 semiclassical series expansion or potential $\hbar^{-1}$ series expansion at large $\hbar$. 
 Instead, we investigate the arithmetic properties of the series expansion of the wave function 
$\Psi(z,\la,\tilde \ka, \hbar)$ modulo $\hbar^{-s}$ when $\hbar^{-1}$ is a prime number.

Interestingly, it turns out that the identification $\hbar^{-1}=p$ naturally emerges in $p$-adic quantum mechanics \cite{FO,BH}, 
where $z,\la \in \mathbb F_p$ and $\F_p$ is the field with $p$ elements.
 In this context, the $p^s$-truncation procedure for the Calogero-Moser system corresponds 
 to selecting particular polynomials in $z,\la$ in the modulo $\hbar^{-s}$-series expansion
  of the wave function $\Psi(z,\la,\tilde \ka, \hbar)$. As we take the limit $s\to \infty$, we obtain 
 an $\hbar^{-1}$-adic version of the quantum mechanical wave function, while the finite
  $s$ case corresponds to a specific approximation of the $p$-adic limit.

\subsection{Semiclassical limit modulo $p^s$}
\label{sec 5.2}

Recall the semiclassical limit of the Matsuo-Chered\-nik correspondence which is known as
classical-quantum duality in the integrability framework.
In the semiclassical limit the integral representations for the solutions of the KZ equations get localized at the saddle points of the 
master function $\Phi(t, z,\la, \ka)$ with respect to the $t$-variables.
The saddle point equations for the master function coincide with the system of the  Bethe Anzatz equations in the Gaudin 
model for the Bethe roots $t^0=(t_{i,j}^0)$, see \cite{RV}. 
The vector weights $V(t,z)$ in the integral representations for solutions of  the KZ equations
become the Bethe vectors - eigenfunctions of  Gaudin Hamiltonians $H_i$\,, 
\bea
H_i(z,\la,\ka) V(t^0,z)= m_i(t^0,z,\la,\ka) V(t^0,z), \quad m_i =\frac{\der \ln \Phi(t^0,z,\lambda,\ka)}{\der z_i}\,, 
\quad
i=1,\dots,n.
\eea

At the Calogero-Moser side as $\hbar \to  0$, we get the classical 
model with the phase space variables $z=(z_1,\dots,z_n)$, $m=(m_1,\dots,m_n)$.
Given $\la=(\la_1,\dots,\la_n)$, the symmetric powers  $\la_1^k + \dots +\la_n^k$, $k=1,2,\dots$, give 
the values of all standard commuting Hamiltonians of the classical Calogero-Moser system. 
There are two natural types of Lagrangian submanifolds in the classical phase space:
a submanifold $L_1$ of the first type is defined by fixing all the $z$-coordinates, while 
a submanifold $L_2$ of the second type is defined by fixing all the $\la$-parameters.
If we evaluate the momenta $m_i$ of Calogero-Moser particles at points of
  $L_1 \cap L_2$\,,\ they coincide 
with eigenvalues $m_i= \frac{\der \ln \Phi(z, t^0,\la,\ka)}{\der z_i} $  of the Gaudin Hamiltonians 
evaluated at the solutions $t^0$  of the Bethe Ansatz equations, see \cite{MTV2, GK, GZZ}.

The truncation of the semiclassical limit of  solutions of KZ equations 
in terms of the truncation of the Bethe Ansatz equations modulo $p$ has been considered in \cite{V2} 
where the algebra $\frak{sl}_2$  and its representations 
are considered over the field $\F_p$. It was shown that solutions of the Bethe Ansatz equations over $\F_p$ give
eigenvectors of the Gaudin Hamiltonians over $\F_p$.

We check the consistency of the semiclassical limit with the truncation procedure for the
simplest  case  $n=2$. That is, we evaluate the momenta $(m_1,m_2)$ of particles
at the intersection of $L_1\cap L_2$, where $L_1$ is defined by fixing $(z_1,z_2)$ and $L_2$ is defined by fixing $(\la_1,\la_2)$,
and compare them with the eigenvalues of the Gaudin Hamiltonians evaluated at the solutions
of the Bethe equations. 

For $n=2$, the phase space has  coordinates $(z_1,z_2,m_1,m_2)$. The intersection of the two Lagrangian submanifords
reads as 
\bean
\label{inter}
 m_1+ m_2 = p(\la_1+\la_2)  ,\qquad
    m_1^2+m_2 ^2 - \frac{2\ka^2}{(z_1-z_2)^2}= p^2 (\la_1^2+\la_2^2) \,,
\eean
 where the rescaling $\la \rightarrow p\la$ is taken into account.

The Bethe ansatz equation associated with the master function
\bean
\label{mF}
\Phi(t,z,p\la, \ka) 
=
e^{p\la_1(z_1+z_2)  + p(\la_{2}-\la_1)t} (z_1-z_2)^\ka
(t-z_1)^{-\ka}(t-z_2)^{-\ka}
\eean
is the equation
\bean
\label{ba}
 p(\la_{2}-\la_1)  - \frac \ka{t-z_1} - \frac \ka{t-z_2} = 0.
 \eean
If $t$ is a solution of the Bethe ansatz equation, then
\bea
m_1 = p\la_1 +\frac{\ka}{z_1-z_2} +\frac{\ka}{t-z_1}\,,
\qquad 
m_2 = p\la_1 +\frac{\ka}{z_2-z_1} +\frac{\ka}{t-z_2}\,,
\eea
are eigenvalues of the corresponding Gaudin Hamiltonians $H_1(z,\la,\ka)$ and $H_2(z,\la,\ka)$, respectively, see \cite{RV}. 
It was shown in \cite{MTV2} that if $t$ is a solution of the Bethe ansatz equation \eqref{ba}, then 
$(z_1,z_2,m_1,m_2)$ satisfy equations \eqref{inter}.
Indeed,
\bea
m_1+m_2 = 2p\la_1 + \Big( \frac \ka{t-z_1} + \frac \ka{t-z_1} -p(\la_{2}-\la_1)  \Big) + p(\la_{2}-\la_1)  
= p(\la_1+\la_2),
\eea
\bea
&&
 m_1^2+m_2 ^2 - \frac{2\ka^2}{(z_1-z_2)^2}
 = \Big(p\la_1 +\frac{\ka}{z_1-z_2} +\frac{\ka}{t-z_1}\Big)^2
+ \Big( p\la_1 +\frac{\ka}{z_2-z_1} +\frac{\ka}{t-z_2}\Big)^2
\\
&&
- \ \frac{2\ka^2}{(z_1-z_2)^2}
=
2p^2\la_1^2 + 2p\la_1 \Big(\Big( \frac \ka{t-z_1} + \frac \ka{t-z_1} -p(\la_{2}-\la_1)  \Big) +p(\la_{2}-\la_1)\Big)
\\
&&
+\ \ka^2\Big(\frac 2{(z_1-z_2)(t-z_1)} +\frac1{(t-z_1)^2}
+ \frac 2{(z_2-z_1)(t-z_2)} +\frac1{(t-z_2)^2}\Big)
\\
&&
= \ 
2p^2\la_1\la_2
+\ka^2\Big(\frac 2{(t-z_1)(t-z_2)} +\frac1{(t-z_1)^2} +\frac1{(t-z_2)^2}\Big)
\\
&&
= \ 
2p^2\la_1\la_2
+\Big(\frac \ka{t-z_2} +\frac\ka{t-z_1}\Big)^2
=  
2p^2\la_1\la_2 + p^2(\la_2-\la_1)^2 = p^2(\la_1^2+\la_2^2) .
\eea
In this calculation, the Bethe Ansatz equation was utilized and the expression $\frac \ka{t-z_1} + \frac \ka{t-z_1}$ was repeatedly replaced with
 $p(\la_{2}-\la_1)$.  Suppose we were able to identify a value for $t$ that solves the Bethe Ansatz equations up to a high power of $p$. In that case, the calculated values for $m_1$ and $m_2$ corresponding to that particular $t$ would fulfill equations \eqref{inter} up to a high power of $p$.
Here is how we solve the Bethe Ansatz equation \eqref{ba} modulo a high power of $p$. Equation \eqref{ba} is a quadratic equation
with solutions
\bean
\label{root}
&&
\\
\notag
t_{\pm} (z,\la,\ka)
&=&
 \frac{z_1+z_2}2 + \frac\ka {p(\la_2-\la_1)} \pm \frac\ka {p(\la_2-\la_1)}\sqrt{1+ \Big(p\frac{(\la_2-\la_1)(z_1-z_2)}{2\ka}\Big)^2}
\\
\notag
&=&
\frac{z_1+z_2}2 + \frac\ka {p(\la_2-\la_1)} \pm \frac\ka {p(\la_2-\la_1)}\sum_{k=0}^\infty \binom{\frac12}{k} \Big(p\frac{(\la_2-\la_1)(z_1-z_2)}{2\ka}\Big)^{2k}\,.
\eean
The root $t_-(z,\la,\ka)$ is of order 0 and the root $t_+(z,\la,\ka)$
 is of order $-1$.  A truncation of the series in \eqref{root} modulo a high power of $p$ gives us a
solution of the Bethe Ansatz equations modulo a high power of $p$. Then the corresponding 
$m_1(t_{\pm}(z,\la,\ka),z,\la,\ka)$ and $m_2(t_{\pm}(z,\la,\ka),z,\la,\ka)$ solve equations \eqref{inter} modulo a high power of $p$.

\section{Discussion}
\label{sec 6}

In this note, by utilizing the Matsuo-Cherednik duality, we have obtained a $p^s$-truncated version of the Calogero-Moser eigenfunctions and explored their properties. The analysis conducted for the simplest rational Calogero-Moser system can be extended to encompass more general systems such as the Ruijsenaars-Schneider or Toda systems. To address the Ruijsenaars-Schneider model, it is necessary to initiate the truncation procedure starting from the truncated version of the integral representation for hypergeometric solutions to the qKZ equations. Additionally, it would be interesting
 to investigate the commutativity between the bispectral duality and the $p^s$-truncation method.
 We have investigated the $s\rightarrow \infty$ limit of the $p^s$ truncation 
to the $p$-adic case for the $n=2$ case only. More general analysis for $n>2$ is required.

It is known for a while \cite{OP,E} that the Calogero-Moser eigenfunctions are closely
related to the representations of Lie groups, for instance to representations of the $SU(N)$ group 
for the $N$-particle Calogero-Moser system. In particular, useful representations
of  wave functions  for trigonometric and rational Calogero system have the form 
\bea
\Psi(\lambda,z,\ka)= Tr_{\lambda} V_{\ka}e^{itz},
\eea
where $V_{\ka}$ is the intertwiner between an $SL(N)$ Verma module  $M_{\lambda}$ and a
finite dimensional representation while $t$ belongs to Cartan subalgebra, see \cite{EK,FV}.
It would be interesting to get the similar representation for the truncated
Calogero-Moser eigenfunctions 
in terms of twisted intertwiner for the group in characteristic $p$. Probably the results from \cite{BFG}
or procedure of  Verma module truncation developed in \cite{V1} could be useful.

Another point concerns the relation of the Calogero-Moser operator with the rational spherical
Cherednik algebra $H_c$ parametrized by  $c$ and defined as follows   \cite{Ch2}.
Consider the operators $x_i$ which multiply polynomials $p(x_1\dots x_n)$ by $x_i$
and the Dunkl operators acting as 
\bea
D_i(p)=\frac{\partial p}{\partial x_i} + c \sum_{j\neq i} \frac{s_{ij}p-p}{x_i-x_j},
\eea
$[D_i,D_j]=0$. The  operators  $(x_i,y_i,s) \in H_c$ 
obey the
following relations 
\bea
[x,x']=0, \qquad  sxs^{-1}=s(x), \qquad
[y,y']=0, \qquad sys^{-1}=s(y),
\eea
\bea
[y,x]=<y,x> -c\sum_{s\in S}<y,\alpha_s><\alpha_s^{v},x> s,
\eea
where $(x,x')\in V^{*}=\on{Span}(x_1,\dots x_n)$, $(y,y')\in V = \on{Span}(D_1,\dots D_n)$, 
$s\in S_n$. When $c=\frac{r}{q}\,,$\  the rational
Cherednik algebra has a finite-dimensional polynomial representation $L_{r/q}$ \cite{BEG}.
To make the link to the Calogero-Moser operator we have to consider the spherical
subalgebra of $H_c$ defined as 
$eH_ce$ where
\bea
e=\frac{1}{|S_n|}\sum_{s\in S_n} s.
\eea
The finite-dimensional polynomial representation of spherical subalgebra is $eL_{r/q}$
and corresponding polynomials are eigenfunctions of
the Calogero-Moser operator $H_{CM}=\sum_{i=1}^N y_i^2$. 

It would be interesting to clarify if 
$p^s$ truncated polynomial Calogero-Moser eigenfunctions
can be related  to the truncation of the finite-dimensional polynomial
representation of spherical subalgebra of the rational Cherednik algebra.
Notice also that at the rational Calogero-Moser coupling $c=\ka=\frac{r}{q}$ the graded dimensions
of the space of polynomial Calogero-Moser eigenfunctions with the eigenvalue $E=0$ were identified 
with the HOMFLY invariants of the $T_{r,q}$ torus knots \cite{GORS, Ch3,EGL}. 
It would be interesting to determine if the $p^s$ truncation procedure for the 
torus knot invariants is available.

The quantum Calogero-Moser system and its generalizations appear in many physical systems
hence our findings are of some interest for them. Let us restrict ourselves by  short comments concerning 
the relation with the perturbed topological gauge theories in two and three dimensions and 
$\Omega$-deformed supersymmetric non-Abelian gauge theories in four and five dimensions. The 
quantum trigonometric $N$-particle Calogero-Moser system can be derived from 
quantum $SU(N)$ 2D Yang--Mills theory on the cylinder with the Wilson line 
in the finite-dimensional representation \cite{GN1}. Similarly the quantum $N$-body
trigonometric Ruijsenaars-Schneider system can be derived from the $SU(N)$ Chern-Simons
theory on $T^2\times R^1$ with two Wilson lines inserted properly \cite{GN2},
(see, \cite{G} for the review).
This realization  via the perturbed topological theories provides
the representation of the Calogero-Moser eigenfunctions in terms of 
twisted intertwiners.
The truncation procedure discussed in this note presumably can be formulated
at the level of  Lagrangians for such perturbed topological gauge theories.

The $N$-particle quantum Calogero-Moser system describes the wave function of the 
corresponding surface defect in the $\mathcal{N}=2$ SUSY  $SU(N)$ $\Omega$-deformed gauge theory
with the massive adjoint matter in four dimensions
where $\ka$ parameter is identified with the mass of the adjoint matter \cite{NS1,NS2,NS3}. Similarly the 
quantum Ruijsenaars-Schneider system corresponds to the lifting of four-dimensional
system to five-dimensional gauge theory with one compact coordinate. 
The $\Omega$ - deformation of SYM theory involves two 
equivariant parameters $\epsilon_1,\epsilon_2$ with respect 
to  two $U(1)$ rotations in $\R^4$. In the Nekrasov-Shatashvili limit
one considers $\epsilon_1=0$ and it turns out that $\epsilon_2$ becomes
the effective Planck constant $\hbar$ in the quantum Calogero-Moser
or Toda model for the surface defect, see, for instance, the rigorous
derivation  in \cite{NT,JLN}.

The truncation procedure we have discussed could have 
at least two interesting applications for
$\mathcal{N}=2$ SUSY gauge 
theories. First, from the 
quantum Calogero-Moser or Toda wave functions we can derive prepotentials 
which describe  low-energy effective actions of SUSY gauge theories
depending on the equivariant parameters of $\Omega$-deformation \cite{N}. It would 
be interesting to get the $p^s$ truncated version of the prepotentials
somewhat similarly to the derivation of the $p^s$ truncated version
of superpotential or vertex function discussed in \cite{SmV1,SmV2}.

Secondly, 
as we have mentioned above if the 
Planck constant has the special form $\hbar=\frac{1}{p}$ where $p$ is prime number
the interesting $\hbar ^{-s}$ truncation is possible for the wave function. In
the $\Omega$-deformed $\mathcal{N}=2$ SUSY gauge theories in the Nekrasov-Shatashvili limit 
$\hbar=\epsilon_2$ hence if $\epsilon_2= \frac{1}{p}$ with $p$ prime
we can develop the $\epsilon_2^{-s}$ truncation procedure. In the $s\rightarrow \infty$ limit
we can presumably get the proper $p$-adic limit of the $\Omega$-deformed 
SUSY theories with the surface defects.

\end{document}